\documentclass[11pt]{article}

\usepackage[T1]{fontenc}
\usepackage[utf8]{inputenc}
\usepackage{lmodern}
\usepackage[margin=1in]{geometry}
\usepackage{amsmath,amssymb,amsthm}
\usepackage{microtype}
\usepackage[
  backend=biber,
  style=numeric,
  natbib=true,
  maxalphanames=3,
  minalphanames=3,
  maxbibnames=99
]{biblatex}

\usepackage[dvipsnames]{xcolor}
\usepackage[
    colorlinks=true,
    linkcolor=MidnightBlue,
    citecolor=BrickRed,
    urlcolor=PineGreen
]{hyperref}

\newtheorem{theorem}{Theorem}

\newtheorem{lemma}{Lemma}
\newtheorem{proposition}{Proposition}
\theoremstyle{definition}
\newtheorem{definition}{Definition}
\theoremstyle{remark}

\theoremstyle{plain}
\newcommand{\E}{\mathbb{E}}

\newcommand{\OPT}{\operatorname{OPT}}
\newcommand{\ON}{\operatorname{ON}}
\newcommand{\ALG}{\operatorname{ALG}}
\newcommand{\Gap}{\operatorname{Gap}}
\newcommand{\Z}{\mathbb Z}
\newcommand{\R}{\mathbb R}

\usepackage{tikz}
\usetikzlibrary{arrows.meta,positioning,calc,patterns,decorations.pathreplacing}
\usepackage[font=small,labelfont=bf]{caption}
\definecolor{blk0}{RGB}{166,206,227}
\definecolor{blk1}{RGB}{253,191,111}
\definecolor{blk2}{RGB}{178,223,138}
\definecolor{blk3}{RGB}{251,154,153}
\definecolor{blk4}{RGB}{202,178,214}
\definecolor{good}{RGB}{51,140,44}

\usepackage{authblk}
\title{A Nearly Tight Lower Bound for Matroid Intersection \\ 
Prophet Inequalities}

\author[1,2]{Dimitris Fotakis}
\author[3]{Charalampos Platanos}
\author[1,2]{Thanos Tolias}

\affil[1]{\small{National Technical University of Athens, Greece}}
\affil[2]{\small{Archimedes RU, Athena RC, Greece}}
\affil[3]{\small{Massachusetts Institute of Technology, USA}}

\affil[ ]{\small{\texttt{{ \href{mailto:fotakis@cs.ntua.gr}{fotakis@cs.ntua.gr} \quad
\href{mailto:platanos@csail.mit.edu}{platanos@csail.mit.edu} \quad
\href{mailto:thanostolias@mail.ntua.gr}{thanostolias@mail.ntua.gr}}}}}
\date{}
\begin{document}

\maketitle

\renewcommand\thefootnote{}\footnotetext{This work has been partially supported by project MIS 5154714 of the National Recovery and Resilience Plan Greece 2.0 funded by the European Union under the NextGenerationEU Program.

Charalampos Platanos is supported by the Paris Kanellakis Fellowship.}

\begin{abstract}
We study prophet inequalities under intersections of $q$ partition matroids, where an online algorithm irrevocably selects elements with independent nonnegative values drawn from known distributions and revealed in an adversarial order. We prove an $\Omega(q/\log q)$ lower bound on the competitive ratio. Together with the known $O(q)$ upper bounds, this resolves, up to a logarithmic factor, the optimal dependence on $q$, an open question posed by \citet*{CCFPW22} and \citet*{SVW23}. Our construction also yields an $\Omega(d/\log d)$ lower bound for $d$-single-minded auctions, where buyers request fixed bundles of at most $d$ unit-capacity items. Our construction and analysis build on the \emph{big-decisions-first} framework of \citet*{RS} (STOC 2026).

Previous lower bounds for this problem hold even for i.i.d.\ Bernoulli values, whereas our instance uses nonidentical distributions. For i.i.d.\ Bernoulli values, we prove a $q^{1-o(1)}$ lower bound, improving the previous $q^{1/2+\Omega(1/\log\log q)}$ of \citet*{SVW23}. To this end, we construct $p^p$ pairwise qualitatively independent partitions into $p$ blocks on a ground set of only $p^{2+o(1)}$ points. 
\end{abstract}

\section{Introduction}

Prophet inequalities quantify the cost of making irrevocable decisions under uncertainty. An online algorithm observes independent nonnegative values drawn from known distributions and revealed in a fixed adversarial order. It must immediately accept or reject each value while maintaining feasibility, and competes with a prophet who knows all realizations in advance. The classical single-choice problem admits an optimal competitive factor of $2$, due to \citet*{KS77} (also crediting D.J.H. Garling) and independently by \citet*{Cahn}. \citet*{KW12} extended this guarantee to arbitrary matroid constraints. Many allocation problems, however, require a decision to satisfy several constraints simultaneously. Intersections of partition matroids already model this situation: each partition imposes capacity constraints on its blocks, and every accepted element consumes capacity in one block of each partition.

For intersections of $q$ matroids, \citet*{KW12} give an $O(q)$ guarantee, while \citet*{FSZ16} obtain $e(q+1)$. For partition matroids with capacity one per block, \citet*{CCFPW22} improve this guarantee to $q+1$. On the lower-bound side, the $\Omega(\sqrt q)$ bound of \citet*{KW12} was improved to $q^{1/2+\Omega(1/\log\log q)}$ by \citet*{SVW23}, who use product dimension to represent the same feasible family with fewer partition matroids. Both lower bounds hold even in the more restricted setting of unit-capacity partition matroids with i.i.d.\ Bernoulli values.

\citet*{CCFPW22} explicitly ask for the tight competitive ratio for intersections of partition matroids; \citet*{SVW23} again highlight this question and ask whether asymptotically better algorithms exist. We resolve the dependence on $q$ up to a logarithmic factor by proving an $\Omega(q/\log q)$ lower bound.

For an instance $I$, write $\Gap(I)$ for the ratio of the expected offline optimum to the largest expected value achievable by any online algorithm, allowing randomization.

\begin{theorem}
\label{thm:partition-lower-bound}
For every integer $q\ge2$, there is a finite prophet inequality instance $I_q$ whose feasible sets are the common independent sets of $q$ partition matroids, each with capacity one per block, such that
\[
  \Gap(I_q)\ge\frac{q}{4\log\bigl(1+(e-1)q\bigr)} .
\]
\end{theorem}

The same $\Omega(q/\log q)$ lower bound therefore applies to intersections of arbitrary matroids. Our instance uses independent, nonidentically distributed values; for i.i.d.\ Bernoulli values we prove a nearly linear bound, Theorem~\ref{thm:iid} below.

Intersections of unit-capacity partition matroids admit a natural interpretation as single-minded auctions. A seller has one copy of each item, and each buyer requests a fixed, known bundle, with an independent nonnegative value for receiving the entire bundle and zero value otherwise. Introducing one item for each block of each partition turns an element into a buyer requesting the $q$ items corresponding to its blocks. Conversely, when items are divided into $d$ pools and every bundle uses at most one item from each pool, feasibility is the intersection of $d$ unit-capacity partition matroids. These translations preserve values, arrival order, and the online and offline problems. We establish the following equivalent formulation of Theorem~\ref{thm:partition-lower-bound}.

\begin{theorem}
\label{thm:single-minded-lower-bound}
For every integer $d\ge2$, there is a finite single-minded instance $I_d$ with one copy of each item and demand bundles of size at most $d$ such that
\[
  \Gap(I_d)\ge\frac{d}{4\log\bigl(1+(e-1)d\bigr)} .
\]
\end{theorem}

For $d$-single-minded buyers, \citet*{CCFPW22} give a $(d+1)$-competitive mechanism based on static anonymous item prices. Our lower bound shows that this dependence on $d$ is optimal within a logarithmic factor, even for unrestricted online algorithms. An $\Omega(d/\log d)$ lower bound was previously known in the broader combinatorial assignment model, with random valuations over multiple bundles~\cite{MSV}; our construction establishes it with fixed single-minded demands and independent scalar values.

Our construction builds on the big-decisions-first framework of \citet*{RS}. We present the analysis in the single-minded formulation, where these ideas are most naturally expressed.

\paragraph{I.i.d.\ Bernoulli values.}
The lower bounds of \citet*{KW12} and \citet*{SVW23} use an incompatible-group construction originating in Babaioff, Immorlica, and Kleinberg \cite{BIK07}. In the form used by Kleinberg and Weinberg, the instance consists of \(N=p^p\) disjoint groups of \(p\) elements with i.i.d. Bernoulli values of mean \(1/p\), and a set is feasible if and only if it is contained in a single group. Its gap is $\Omega(p)$, so the question is how few unit-capacity partition matroids can express its feasibility constraint. \citet*{KW12} use about $p^2$, and \citet*{SVW23} reduce this to $p^{2-\Omega(1/\log\log p)}$ via the product dimension of a disjoint union of cliques and the covering families of \citet*{AA20}. We reduce it to $p\,e^{O(\sqrt{\log p})}$ and obtain a nearly linear lower bound for i.i.d.\ Bernoulli values. This instance cannot yield a gap larger than \(O(q/\log q)\). Representing its feasibility constraint requires at least \(\lceil\log_2(2p^p)\rceil\) unit-capacity partition matroids \cite{AA20}. Theorem \ref{thm:iid} is within a factor \(e^{O(\sqrt{\log q})}\) of this limit.

\begin{theorem}
\label{thm:iid}
For every integer $q\ge2$, there is a finite prophet inequality instance $I_q$ with i.i.d.\ Bernoulli values whose feasible sets are the common independent sets of $q$ partition matroids, each with capacity one per block, such that
\[
  \Gap(I_q)\ge q\,e^{-O(\sqrt{\log q})} .
\]
The bound holds for every arrival order.
\end{theorem}

The partition matroids come from qualitatively independent partitions~\cite{GKV93,GKV94}: partitions such that every block of one meets every block of the other.

\begin{theorem}\label{thm:qi}
    For every sufficiently large integer \(p\), there exist \(p^{3p}\) pairwise qualitatively independent partitions into \(p\) blocks on a ground set \(\Omega\) of at most \(p^2e^{O(\sqrt{\log p})}\) points. Moreover, \(\Omega\) admits a partition into sets of size \(p\), each meeting every block of every partition in exactly one point.
\end{theorem}

Beyond its application to prophet inequalities, our construction shows that the product dimension of \(p^p\) disjoint cliques of size \(p\) is \(p^{1+o(1)}\).

\subsection{Other related work}

\paragraph{Single-minded allocation and pricing.}
For fixed item multiplicity $B$, \citet*{FPT25} obtain an $O(d^{1/B})$-competitive static anonymous bundle-pricing mechanism and an $\Omega((d/\log d)^{1/(B+1)})$ lower bound using qualitatively independent partitions. \citet*{CMT19} study bundle pricing for interval and tree-path demands.

\paragraph{General feasibility constraints.}
The incompatible-group construction originating in \citet*{BIK07} and used for prophet inequalities by \citet*{CHMS10} gives an $\Omega(\log n/\log\log n)$ lower bound for downward-closed constraints on $n$ elements. \citet*{Rub16} obtains a nearly tight $O(\log n)$ guarantee for binary values, whereas \citet*{RS} establish a nearly tight $\widetilde\Omega(\log^2 n)$ lower bound using multiple positive value scales (here the tilde hides factors polynomial in $\log\log n$).

\paragraph{Qualitatively independent partitions.}
Families of pairwise qualitatively independent partitions were studied by \citet*{GKV93,GKV94} in connection with Sperner capacity; in design theory, the shift matrices behind our construction are difference covering arrays~\cite{Yin03,CD07}, usually studied with few columns. The connection to prophet inequalities under intersections of partition matroids was established by Kleinberg and Weinberg \cite{KW12} and further developed by Saxena, Velusamy, and Weinberg \cite{SVW23}.

\subsection{Technical overview}

Our construction follows the \emph{big-decisions-first} principle of \citet*{RS}: the offline allocation benefits from prioritizing buyers with larger positive values, while the online algorithm encounters buyers with smaller positive values first. We implement this conflict using $d$ layers with overlapping demand bundles, so that accepting early buyers consumes items needed by later buyers with larger values. Take disjoint item pools $M_\ell$ of sizes $r_\ell=8^\ell$. A layer-$\ell$ buyer demands one item from each of $M_\ell,\ldots,M_d$ and has value $1/r_\ell$ when active, and zero otherwise. All activations are independent, with a common probability; buyer multiplicities make the expected total value of each layer one. Layers arrive in the order $d,d-1,\ldots,1$, so accepting smaller-value buyers can obstruct later buyers of larger value. Each buyer has a two-point distribution; its positive value depends on the layer.

The offline allocation processes layers in the opposite order, $1,2,\ldots,d$. Since pool sizes grow geometrically, selections from higher-value layers occupy only a small fraction of the items in each subsequent pool. A greedy allocation obtains constant expected value from every layer, yielding $\mathbb E[\mathrm{OPT}]\ge d/4$.

For any online algorithm, let $Y_\ell\in[0,1]$ be its reward from layer $\ell$ and $Q_\ell=\sum_{j>\ell}Y_j$ its preceding reward. Earning $Y_j$ occupies at least a $Y_j$ fraction of $M_j$. Conditional on the history $\mathcal H_\ell$ before layer $\ell$, counting the bundles that avoid occupied items gives
\[
  \mathbb E[Y_\ell\mid\mathcal H_\ell]
  \le \prod_{j>\ell}(1-Y_j)\le e^{-Q_\ell}.
\]
We then use an exponential-potential argument similar in flavor to the one used in the proof of \cite[Lemma 1]{MSV}. More specifically, we prove \( \mathbb E[e^{Q_{\ell-1}}\mid\mathcal H_\ell] \le e^{Q_\ell}+e-1.\)
Iterating this inequality and applying Jensen’s inequality gives $\mathbb E[\mathrm{ALG}]\le\log(1+(e-1)d)$, completing the separation.

\paragraph*{Qualitatively independent partitions.} We label each partition by a binary word and specify its blocks through cyclic shifts on the common sets of size \(p\). Two partitions are qualitatively independent precisely when the differences between their corresponding shifts cover every residue modulo \(p\). This cyclic-shift viewpoint is standard (see the difference covering arrays over $\mathbb Z_p$~\cite{Yin03}). The challenge is to satisfy this condition for every pair of distinct words while using few shifts.

A direct construction uses each bit separately: multiply it by each weight \(1,\ldots,\lfloor p/2\rfloor\), and include both signs as shifts. Whenever two words differ in a bit, the corresponding shift differences supply every nonzero residue modulo \(p\); an additional zero shift supplies residue zero. This works, but requires two shifts for every bit–weight pair.

To use fewer shifts, we group weighted bits according to sums of lattice indices. Let \(A\subseteq\mathbb Z^r\) index the bit positions, and let \(B\subseteq\mathbb Z^r\) consist of \(\lfloor p/2\rfloor\) points indexing the weights. Write \(x_a\) for the bit at position \(a\in A\), and assign each \(b\in B\) a distinct weight \(\lambda(b)\in\{1,\ldots,\lfloor p/2\rfloor\}\). For each \(t\) in the sumset \(A+B=\{a+b:a\in A,\ b\in B\}\), we form the shift
\[
u_t(x)=\sum_{\substack{a\in A,\ b\in B\\a+b=t}}\lambda(b)x_a\pmod p.
\]
Including both signs of these shifts and the zero shift uses only \(2|A+B|+1\) shifts, instead of \(2|A||B|+1\). The difficulty is that contributions to the same sum can cancel when comparing two words.

A geometric isolation property ensures that every required residue survives these cancellations. We choose \(B\) in convex position. If \(D\subseteq A\) is the nonempty set of positions where two words differ, then, for every \(b\in B\), some \(a\in D\) makes \(a+b\) uniquely representable in \(D+B\). Indeed, choose a linear functional uniquely maximized at \(b\) on \(B\), and then choose \(a\) maximizing the same functional on \(D\). At the index \(a+b\), the difference between the two weighted sums is therefore exactly \(\pm\lambda(b)\). Both signs, together with the zero shift, give all residues modulo \(p\).

Finally, we take \(A\) to be a lattice box large enough to index the required bits, and choose \(B\) from a sphere layer inside another lattice box, closely resembling the construction of Behrend in \cite{Beh46}. The sphere guarantees convex position, while the enclosing boxes control the size of \(A+B\). Choosing the dimension \(r=\Theta(\sqrt{\log p})\) gives \(|A+B|\le p\exp(O(\sqrt{\log p}))\), yielding Theorem~\ref{thm:qi}.

\section{Model, preliminaries, and notation}
\label{sec:model}

We consider instances of $d$-single-minded prophet inequalities. More specifically, an instance $I=(M,B,(S_b,\mathcal D_b)_{b\in B},\pi)$ consists of a finite set of unit-capacity items $M$, a finite set of buyers $B$, a nonempty demand bundle $S_b\subseteq M$ with $|S_b|\leq d$ and a value distribution $\mathcal D_b$ for each buyer, and a fixed arrival order $\pi$. Values $V_b\sim \mathcal D_b$ are independent. Each buyer $b$ is single-minded: for every $T\subseteq M$, its valuation is
\[
  \nu_b(T)=
  \begin{cases}
    V_b, & S_b\subseteq T,\\
    0,   & \text{otherwise}.
  \end{cases}
\]
The family of feasible buyer sets is
\(
  \mathcal F
  =
  \{A\subseteq B:
    S_b\cap S_{b'}=\varnothing
    \text{ for all distinct }b,b'\in A\}.
\)

An online algorithm $\mathcal{A}$ knows $I$ in advance. Upon buyer $b$'s arrival,
it observes $V_b$ and irrevocably accepts or rejects $b$, maintaining
feasibility. Its reward is
$\operatorname{ALG}_{\mathcal A}
=\sum_{b\in A_{\mathcal A}}V_b$, where $A_{\mathcal A}$ is its accepted set. For a value realization $\mathbf v=(v_b)_{b\in B}$, define
\[
  \operatorname{OPT}(\mathbf v)
  =\max_{A\in\mathcal F}\sum_{b\in A}v_b,
  \qquad
  \operatorname{ON}(I)
  =\sup_{\mathcal A}\mathbb E[\operatorname{ALG}_{\mathcal A}].
\]
Let $\mathbf V=(V_b)_{b\in B}$, write $\mathrm{OPT}=\mathrm{OPT}(\mathbf V)$, and define
\[
  \operatorname{Gap}(I)
  =\frac{\mathbb E[\operatorname{OPT}]}{\operatorname{ON}(I)}.
\]
An algorithm is \emph{$\rho$-competitive}, for $\rho\ge1$, if
$\mathbb E[\operatorname{ALG}_{\mathcal A}]
\ge \mathbb E[\operatorname{OPT}]/\rho$ for every instance. Finally, an item is \emph{free} if no selected bundle contains it, and
\emph{occupied} otherwise. We write $[d]=\{1,\ldots,d\}$,
use natural logarithms, and assign empty sums and products the
values zero and one, respectively.

\section{Proofs of our main results}
\label{sec:main}
In this section we prove Theorems \ref{thm:partition-lower-bound} and \ref{thm:single-minded-lower-bound}.
Fix an integer $d\ge2$. We define $r_\ell=8^\ell$ for $\ell \in [d]$, $P=\prod_{j=1}^d r_j=8^{d(d+1)/2}$ and $p=1/P$. We construct $I_d$ below and then prove that its online value is at most
$\log(1+(e-1)d)$, whereas its expected offline value is at least $d/4$.

\subsection{Construction}
\label{subsec:construction}

Take pairwise disjoint item pools $M_1,\ldots,M_d$ with $|M_\ell|=r_\ell$, thus the item set is $M=M_1\cup\ldots \cup M_d$.
For each $\ell\in[d]$ and every tuple
$(x_\ell,\ldots,x_d)\in M_\ell\times\cdots\times M_d$, introduce
\(
  m_\ell=\prod_{j=1}^{\ell}r_j
\)
distinct buyers demanding the bundle $\{x_\ell,\ldots,x_d\}$.
These buyers constitute \emph{layer $\ell$}, and $x_\ell$ is their
\emph{anchor}. Each layer-$\ell$ buyer $b$ independently has value
\[
  V_b=
  \begin{cases}
    1/r_\ell, & \text{with probability }p,\\
    0,        & \text{with probability }1-p.
  \end{cases}
\]
A buyer is \emph{active} when its value is positive. Layers arrive in
the order $d,d-1,\ldots,1$, with an arbitrary fixed order within each
layer. Every bundle has size at most $d$. At each anchor in $M_\ell$, the number
of buyers is
\(
  m_\ell\prod_{j>\ell}r_j=P,
\)
so the total number of buyers is $P\sum_{\ell=1}^d r_\ell$.
Since at most one buyer can be selected per anchor, a feasible allocation
earns at most $r_\ell/r_\ell=1$ from layer $\ell$, and at most $d$ in
total.

We will use the following standard counting identity. Fix a free anchor in
$M_\ell$, and let $U_j\subseteq M_j$ be the occupied items in each deeper
pool $j>\ell$. The number of buyers at this anchor whose bundles are
entirely free is exactly
\begin{equation}\label{eq:count}
  m_\ell\prod_{j>\ell}(r_j-|U_j|)
  =P\prod_{j>\ell}\left(1-\frac{|U_j|}{r_j}\right).
\end{equation}
Indeed, each deeper coordinate has $r_j-|U_j|$ available choices, and
every resulting bundle has $m_\ell$ buyer copies.

\subsection{Online value}
\label{subsec:online}

\begin{lemma}\label{lem:online}
The instance $I_d$ satisfies
\(
  \ON(I_d)\le\log(1+(e-1)d).
\)
\end{lemma}

\begin{proof}
Fix an arbitrary randomized online algorithm $\mathcal A$.
Let $Y_\ell\in[0,1]$ be its reward from layer $\ell$, and write
$Q_\ell=\sum_{j>\ell}Y_j$ for its reward before that layer arrives.
Let $\mathcal H_\ell$ be the sigma-algebra generated by the algorithm's
independent random seed and all buyer activations in layers $j>\ell$.

At the start of layer $\ell$, every anchor in $M_\ell$ is free, since
previously arriving buyers request items only from deeper pools.
For $j>\ell$, let $U_j$ be the occupied items in $M_j$ at this time.
The positive-value acceptances from layer $j$ occupy $r_jY_j$ distinct
anchors in $M_j$, and hence
\(
  {|U_j|}/{r_j}\ge Y_j.
\)

Every buyer accepted during layer $\ell$ must have been feasible at the
start of the layer. Conditional on $\mathcal H_\ell$, the initially
feasible buyer set is fixed, and all layer-$\ell$ activations remain
independent with probability $p$. By~\eqref{eq:count}, there are
$r_\ell P\prod_{j>\ell}(1-|U_j|/r_j)$ such buyers, each with expected
value $p/r_\ell$. Since $pP=1$, it follows that
\begin{equation}\label{eq:conditional-online}
  \E[Y_\ell\mid\mathcal H_\ell]
  \le\prod_{j>\ell}\left(1-\frac{|U_j|}{r_j}\right)
  \le\prod_{j>\ell}(1-Y_j)
  \le e^{-Q_\ell}.
\end{equation}
For $0\le y\le1$, convexity gives $e^y\le1+(e-1)y$. As $Q_\ell$ is
$\mathcal H_\ell$-measurable, \eqref{eq:conditional-online} yields
\begin{align*}
 \E[e^{Q_{\ell-1}}\mid\mathcal H_\ell] = \E[e^{Q_\ell+Y_\ell}\mid\mathcal H_\ell]
  &\le e^{Q_\ell}
    \bigl(1+(e-1)\E[Y_\ell\mid\mathcal H_\ell]\bigr)\\
  &\le e^{Q_\ell}+e-1.
\end{align*}
Taking expectations and iterating over the $d$ arriving layers, starting
from zero reward, gives
\(
  \E[e^{\ALG_{\mathcal A}}]\le1+(e-1)d.
\)
Finally, Jensen's inequality gives
\[
  \E[\ALG_{\mathcal A}]
  \le\log\E[e^{\ALG_{\mathcal A}}]
  \le\log(1+(e-1)d).
\]
The bound holds for every $\mathcal A$, thus proving the lemma.
\end{proof}

\subsection{Value of the optimal solution}
\label{subsec:offline}

\begin{lemma}\label{lem:offline}
The instance $I_d$ satisfies $\E[\OPT]\ge d/4$.
\end{lemma}

\begin{proof}
Consider the following allocation procedure. Process layers in the order
$1,2,\ldots,d$, and process the anchors of each layer in a fixed order.
At an anchor, expose the values of all buyers anchored there. If the
anchor is free and some active buyer has a completely free bundle,
select the first such buyer according to a fixed tie-breaking order;
otherwise, select no buyer. The procedure maintains feasibility, so its
total reward $G$ satisfies $G\le\OPT$ for every realization.

Before layer $\ell$ is processed, there have been at most
$\sum_{i<\ell}r_i$ selections. Each uses one item of $M_\ell$. Since
\[
  \frac{\sum_{i<\ell}r_i}{r_\ell}
  =\sum_{h=1}^{\ell-1}8^{-h}
  \le\frac17\le\frac14,
\]
at least $3r_\ell/4$ anchors are initially free. Each remains free
until its turn, because a layer-$\ell$ bundle uses no other anchor in
$M_\ell$.

Immediately before processing any free anchor in layer $\ell$, let $K$
be the number of previous selections, including those earlier in the
same layer. Every previous selection uses exactly one item of each
deeper pool $M_j$, $j>\ell$. Since selected bundles are disjoint,
exactly $K$ items are occupied in each such pool. Moreover,
\[
  K\le\sum_{i\le\ell}r_i\le\frac87r_\ell,
  \qquad
  \sum_{j>\ell}\frac1{r_j}\le\frac1{7r_\ell}.
\]
For $j>\ell$, we also have $K/r_j\le1/7<1$. Applying
$\prod_i(1-u_i)\ge1-\sum_i u_i$ for $u_i\in[0,1]$ gives
\begin{equation}\label{eq:availability}
  \prod_{j>\ell}\left(1-\frac K{r_j}\right)
  \ge1-K\sum_{j>\ell}\frac1{r_j}
  \ge1-\frac8{49}
  \ge\frac34.
\end{equation}
This also holds for $\ell=d$, when the product is empty.
By~\eqref{eq:count}, at least $3P/4$ buyer labels at the current anchor
have compatible bundles.

Condition on all activations exposed before the current anchor is
processed. The compatible buyer set is determined by this history.
Different anchors have disjoint buyer labels, so the activations of
buyers at the current anchor remain independent Bernoulli variables
with parameter $p$. If $F$ is the number of compatible labels, the
conditional probability of making a selection is therefore
\begin{equation}\label{eq:success}
  1-(1-p)^F
  \ge1-e^{-pF}
  \ge1-e^{-3/4}.
\end{equation}
Conditional on the history at the start of layer $\ell$, its initially
free anchors form a fixed set of size at least $3r_\ell/4$, and each
remains free until its turn. Applying~\eqref{eq:success} at each such
anchor the expected reward from the layer
is at least
\[
  \frac{3r_\ell}{4}\cdot\frac1{r_\ell}
  \bigl(1-e^{-3/4}\bigr)
  =\frac34\bigl(1-e^{-3/4}\bigr).
\]
Summing over all layers and using $e^{3/4}\ge1+3/4=7/4$, we obtain
\[
  \E[\OPT]\ge\E[G]
  \ge\frac{3d}{4}\bigl(1-e^{-3/4}\bigr)
  \ge\frac{3d}{4}\left(1-\frac47\right)
  =\frac{9d}{28}>\frac d4.
\]
\end{proof}

\begin{proof}[Proof of Theorem~\ref{thm:single-minded-lower-bound}]
By Lemmas~\ref{lem:online} and~\ref{lem:offline},
\[
  \Gap(I_d)=\frac{\E[\OPT]}{\ON(I_d)}
  \ge\frac{d}{4\log(1+(e-1)d)},
\]
as required.
\end{proof}

\subsection{Partition-matroid intersection}
\label{subsec:partition-matroid-intersection}

\begin{proof}[Proof of Theorem~\ref{thm:partition-lower-bound}]
Fix an integer $q\ge 2$, and let $J$ be the single-minded instance supplied
by Theorem~\ref{thm:single-minded-lower-bound} with $d=q$. Let $\mathcal B$ be its
buyer set and $M_1,\ldots,M_q$ the partition of its items guaranteed by that
theorem. Thus, $|S_b\cap M_j|\le 1$ for every buyer $b\in\mathcal B$ and
every $j\in[q]$. We use the standard correspondence with partition
matroids~\citep[Footnote~4]{SVW23}, spelling out the construction below.

For each $j\in[q]$ and each item $x\in M_j$, define
\[
    B_{j,x}:=\{b\in\mathcal B:x\in S_b\}.
\]
Since each demand bundle contains at most one item from $M_j$, the
nonempty sets $B_{j,x}$ are pairwise disjoint. Complete them to a partition
of $\mathcal B$ by adding a singleton block $\{b\}$ for every buyer whose
bundle does not intersect $M_j$. Let $\mathcal N_j$ be the partition
matroid on $\mathcal B$ defined by these blocks, each with capacity one.
Distinct buyers remain distinct ground-set elements even when their demand
bundles coincide.

A set $A\subseteq\mathcal B$ is independent in every $\mathcal N_j$ if
and only if $|A\cap B_{j,x}|\le 1$ for every $j\in[q]$ and $x\in M_j$;
the singleton blocks impose no additional restriction. Since the pools
partition the item set, these inequalities say precisely that no item is
requested by two distinct buyers in $A$. Thus, the common independent sets
of $\mathcal N_1,\ldots,\mathcal N_q$ are exactly the feasible buyer sets
of $J$.

Define a prophet inequality instance $I_q$ on this common ground set by
assigning element $b$ the same random value $V_b$ as in $J$ and using the
same arrival order. This instance is finite, and its values remain
independent. For every value realization, the two instances have identical
feasible sets and objective values, so their offline optima coincide.
Moreover, an online algorithm for either instance can be simulated on the
other by making the same accept-or-reject decisions after the same observed
values. This preserves feasibility and realized reward, also for randomized
algorithms, and hence $\operatorname{ON}(I_q)=\operatorname{ON}(J)$.
Consequently, Theorem~\ref{thm:single-minded-lower-bound} gives
\[
    \operatorname{Gap}(I_q)
    =\operatorname{Gap}(J)
    \ge \frac{q}{4\log\!\bigl(1+(e-1)q\bigr)},
\]
which proves the theorem.
\end{proof}

\section{Qualitatively independent partitions}\label{sec:qi}

This section proves Theorem~\ref{thm:qi} and then deduces
Theorem~\ref{thm:iid}. We construct partitions
using cyclic shifts, then choose integer point sets that keep the
underlying set small. The notation is local to this section; in
particular, $p$ denotes the number of blocks of each partition.

\subsection{Preliminaries}\label{qi-sec:prelim}

Fix integers $p,N\ge2$, and write $\Z_p=\Z/p\Z$.
\begin{definition}\label{qi-def:qi}
Two partitions $\mathcal P=(P_s)_{s\in\Z_p}$ and
$\mathcal Q=(Q_s)_{s\in\Z_p}$ of a finite set $\Omega$ are
\emph{qualitatively independent} if
\[
  P_s\cap Q_t\ne\varnothing
  \qquad\text{for every }s,t\in\Z_p.
\]
\end{definition}

We seek $N$ pairwise qualitatively independent partitions into $p$
blocks, together with a common partition $(C_j)_{j\in J}$ of
$\Omega$ into sets of size $p$, indexed by a finite nonempty set $J$.
We require every block to contain exactly one point of each $C_j$.
Writing $\ell=|J|$, the ground set then has $p\ell$ points and each
block has $\ell$ points. We will use $\Omega=J\times\Z_p$ and
$C_j=\{j\}\times\Z_p$ (Figure~\ref{qi-fig:naming}).

\begin{figure}[hb]
\centering
\begin{tikzpicture}[
  font=\small,
  x=1.05cm,
  y=0.78cm,
  ptc/.style={
    draw,
    rounded corners=5pt,
    fill=white,
    inner sep=0pt,
    minimum width=0.84cm,
    minimum height=0.48cm,
    font=\scriptsize
  },
  pt/.style={
    circle,
    draw,
    inner sep=0pt,
    minimum size=0.5cm
  },
  row/.style={
    draw,
    rounded corners=5pt,
    gray!80
  },
  note/.style={
    font=\footnotesize,
    text=black!80,
    align=center
  }
]
\begin{scope}
  \foreach \j in {1,...,4} {
    \draw[row]
      (-0.5,1-\j-0.4) rectangle (2.5,1-\j+0.4);
    \node[anchor=east, font=\footnotesize]
      at (-0.6,1-\j) {$C_{\j}$};
    \foreach \z in {0,1,2} {
      \node[ptc] at (\z,1-\j) {$(\j,\z)$};
    }
  }
  \node[ptc, fill=gray!30] at (1,-2) {$(3,1)$};
  \draw[decorate, decoration={brace, amplitude=4pt}]
    (-0.5,0.5) -- (2.5,0.5)
    node[midway, above=5pt, note] {$p=3$ points per set};
  \draw[decorate, decoration={brace, amplitude=4pt}]
    (2.65,0.4) -- (2.65,-3.4)
    node[midway, right=6pt, note] {$\ell=4$\\sets};
  \node at (1,-4.4)
    {(a) the ground set $\Omega$, $|\Omega|=p\ell=12$};
\end{scope}
\begin{scope}[xshift=7.6cm]
  \draw[line width=2.4pt, blk0!60!black, line join=round]
    (0,0) -- (1,-1) -- (2,-2) -- (1,-3);
  \foreach \j/\h in {1/0,2/1,3/2,4/1} {
    \draw[row]
      (-0.5,1-\j-0.4) rectangle (2.5,1-\j+0.4);
    \node[anchor=east, font=\footnotesize, gray]
      at (-0.6,1-\j) {$C_{\j}$};
    \foreach \z in {0,1,2} {
      \pgfmathtruncatemacro\s{Mod(\z-\h,3)}
      \node[pt, fill=blk\s] at (\z,1-\j) {$\s$};
    }
  }
  \node[font=\footnotesize, text=blk0!60!black]
    at (1,-3.72) {block $P_0$};
  \draw[decorate, decoration={brace, amplitude=4pt}]
    (-0.5,0.5) -- (2.5,0.5)
    node[midway, above=5pt, note] {blocks $P_0,P_1,P_2$};
  \node at (1,-4.4)
    {(b) a partition $\mathcal P=(P_0,P_1,P_2)$};
\end{scope}
\end{tikzpicture}
\caption{A ground set with $p=3$ and $\ell=4$. Each set $C_j$
contains the points $(j,0),(j,1),(j,2)$. On the right, points are
labeled by their blocks in a partition. Every block contains exactly
one point of each $C_j$; block $P_0$ is traced.}
\label{qi-fig:naming}
\end{figure}

For a positive integer $r$ and finite sets $A,B\subseteq\Z^r$,
write $A+B=\{a+b:a\in A,\ b\in B\}$ for their sumset.
We use the standard Euclidean inner product
$\langle\cdot,\cdot\rangle$ and norm $\|\cdot\|_2$ on $\R^r$.
The construction uses the following geometric observation: a point
of $B$ that uniquely maximizes a linear functional can be isolated
in a sum with any finite nonempty set $D$.

\begin{lemma}\label{qi-lem:exposure}
Let $B,D\subseteq\Z^r$ be finite and nonempty. Suppose $b\in B$
uniquely maximizes $\langle v,\cdot\rangle$ on $B$, for some
$v\in\R^r$. Then some $a\in D$ satisfies
\[
  a'+b'=a+b,\quad a'\in D,\ b'\in B
  \quad\Longrightarrow\quad a'=a\text{ and }b'=b.
\]
\end{lemma}

\begin{proof}
Choose $a\in D$ maximizing $\langle v,a\rangle$. If
$a'+b'=a+b$, then
\[
  \langle v,a'\rangle+\langle v,b'\rangle
  =\langle v,a\rangle+\langle v,b\rangle.
\]
Each term on the left is at most the corresponding term on the
right, so both inequalities are equalities. Uniqueness of the
maximizer on $B$ gives $b'=b$, and hence $a'=a$.
\end{proof}

\begin{lemma}\label{qi-lem:sphere}
If all points of a finite nonempty set $B\subseteq\Z^r$ have the
same Euclidean norm, then every point of $B$ uniquely maximizes a
linear functional on $B$.
\end{lemma}

\begin{proof}
For $b\in B$, use the functional $\langle b,\cdot\rangle$.
For every $b'\in B\setminus\{b\}$, equality of the norms gives
\[
  \langle b,b\rangle-\langle b,b'\rangle
  =\frac12\|b-b'\|_2^2>0.\qedhere
\]
\end{proof}

\subsection{Constructing the partitions}\label{qi-sec:construction}

We label the partitions by binary words
$x=(x_a)_{a\in A}\in\{0,1\}^A$, where $A$ is a finite nonempty
set of positions and $x_a$ is the bit in position $a$.
For each $j\in J$, a word specifies a shift $h_j(x)\in\Z_p$:
block $s$ takes the point $(j,s+h_j(x))$ from $C_j$.
The following lemma states the condition these shifts must satisfy.
All shifts and weighted sums below are computed modulo $p$ (see also Figure~\ref{qi-fig:qi}).

\begin{lemma}\label{qi-lem:shifts}
Let $X$ and $J$ be finite nonempty sets, and suppose that shifts
$h_j(x)\in\Z_p$ are given for $x\in X$ and $j\in J$. If
\[
  \{h_j(x)-h_j(y):j\in J\}=\Z_p
  \qquad\text{for every distinct }x,y\in X,
\]
then there are $|X|$ pairwise qualitatively independent partitions
of $J\times\Z_p$ into $p$ blocks, each block containing exactly one
point of each $C_j=\{j\}\times\Z_p$.
\end{lemma}

\begin{proof}
For $x\in X$, define
\[
  P_s^x=\{(j,s+h_j(x)):j\in J\},
  \qquad \mathcal P^x=(P_s^x)_{s\in\Z_p}.
\]
For each $j$, the map $s\mapsto s+h_j(x)$ is a bijection of
$\Z_p$, so these blocks partition $J\times\Z_p$ and contain
exactly one point of each $C_j$.

For distinct $x,y\in X$, the blocks $P_s^x$ and $P_t^y$ meet in
$C_j$ precisely when
\[
  s+h_j(x)=t+h_j(y),
  \qquad\text{equivalently}\qquad h_j(x)-h_j(y)=t-s.
\]
The hypothesis supplies such an index $j$ for every $s,t$, proving
qualitative independence. It also ensures that the partitions are
distinct, since a partition into $p\ge2$ nonempty blocks cannot be
qualitatively independent of itself.
\end{proof}

\begin{figure}[t]
\centering
\begin{tikzpicture}[
  font=\small,
  y=0.7cm,
  pt/.style={
    circle,
    draw,
    inner sep=0pt,
    minimum size=0.5cm
  },
  row/.style={
    draw,
    rounded corners=4pt,
    gray!80
  }
]
\begin{scope}[x=0.75cm]
  \foreach \j in {1,...,4} {
    \draw[row]
      (-0.5,1-\j-0.4) rectangle (2.5,1-\j+0.4);
    \node[anchor=east, gray] at (-0.6,1-\j) {$C_{\j}$};
    \foreach \z in {0,1,2} {
      \node[pt, fill=blk\z] at (\z,1-\j) {$\z$};
    }
  }
  \node[align=center] at (1,-4.3)
    {(a) $\mathcal P^x$\\shifts $(0,0,0,0)$};
\end{scope}
\begin{scope}[xshift=4.3cm, x=0.75cm]
  \foreach \j/\h in {1/0,2/1,3/2,4/1} {
    \draw[row]
      (-0.5,1-\j-0.4) rectangle (2.5,1-\j+0.4);
    \node[anchor=east, gray] at (-0.6,1-\j) {$C_{\j}$};
    \foreach \z in {0,1,2} {
      \pgfmathtruncatemacro\s{Mod(\z-\h,3)}
      \node[pt, fill=blk\s] at (\z,1-\j) {$\s$};
    }
  }
  \node[align=center] at (1,-4.3)
    {(b) $\mathcal P^y$\\shifts $(0,1,2,1)$};
\end{scope}
\begin{scope}[xshift=8.7cm, x=0.95cm]
  \foreach \j/\h in {1/0,2/1,3/2,4/1} {
    \draw[row]
      (-0.5,1-\j-0.4) rectangle (2.5,1-\j+0.4);
    \node[anchor=east, gray] at (-0.6,1-\j) {$C_{\j}$};
    \pgfmathtruncatemacro\d{Mod(-\h,3)}
    \node[anchor=west, gray, font=\footnotesize]
      at (2.6,1-\j) {$\delta_{\j}=\d$};
    \foreach \z in {0,1,2} {
      \pgfmathtruncatemacro\t{Mod(\z-\h,3)}
      \fill[blk\z]
        (\z-0.38,1-\j-0.25) rectangle (\z,1-\j+0.25);
      \fill[blk\t]
        (\z,1-\j-0.25) rectangle (\z+0.38,1-\j+0.25);
      \draw
        (\z-0.38,1-\j-0.25) rectangle (\z+0.38,1-\j+0.25);
      \node[font=\footnotesize] at (\z,1-\j) {$\z,\t$};
    }
  }
  \node[align=center] at (1,-4.3) {(c) pairs $(s,t)$};
\end{scope}
\end{tikzpicture}
\caption{Two qualitatively independent partitions for $p=3$.
The right panel records the pair of block labels at each point.
All nine pairs occur because the shift differences include every
residue in $\Z_3$ (Lemma~\ref{qi-lem:shifts}). Here
$\delta_j=h_j(x)-h_j(y)$, computed modulo $3$.}
\label{qi-fig:qi}
\end{figure}

Put $\kappa=\lfloor p/2\rfloor$, choose a set $B$ of size
$\kappa$, and fix a bijection
$\lambda:B\to\{1,\ldots,\kappa\}$. The residues
$\{\pm\lambda(b):b\in B\}$ are exactly the nonzero elements of
$\Z_p$. A direct construction uses the two shifts
$\lambda(b)x_a$ and $-\lambda(b)x_a$ for each $(a,b)\in A\times B$,
and one additional shift equal to zero. Indeed, if $x\ne y$, choose
a position $a$ where they differ. For each $b$, the two corresponding
shift differences give both $\lambda(b)$ and $-\lambda(b)$; the
zero shift supplies difference zero. Thus $2|A||B|+1$ sets $C_j$
suffice.

To use fewer sets, place $A$ and $B$ in $\Z^r$ and combine the
terms with the same sum $a+b$. This gives one weighted sum for each
$t\in A+B$, rather than one term for each pair $(a,b)\in A\times B$.
Contributions may cancel. However, if every point of $B$ uniquely
maximizes a linear functional on $B$, Lemma~\ref{qi-lem:exposure}
guarantees a sum with exactly one contribution from the positions
where the two words differ, for each prescribed $b\in B$.

\begin{lemma}\label{qi-lem:measurements}
Let $A,B\subseteq\Z^r$ be finite and nonempty, and suppose every
point of $B$ uniquely maximizes some linear functional on $B$.
Assign weights $\lambda(b)\in\Z_p$ to the points of $B$. For a
binary word $x\in\{0,1\}^A$, define
\[
  u_t(x)=\sum_{\substack{a\in A,\,b\in B\\a+b=t}}
                \lambda(b)x_a
  \qquad(t\in A+B).
\]
For every distinct $x,y\in\{0,1\}^A$ and every $b\in B$, some
$t\in A+B$ satisfies
\[
  u_t(x)-u_t(y)\in\{\lambda(b),-\lambda(b)\}.
\]
\end{lemma}

\begin{proof}
Let $D=\{a\in A:x_a\ne y_a\}$, which is nonempty. Then
\[
  u_t(x)-u_t(y)
  =\sum_{\substack{a\in D,\,b'\in B\\a+b'=t}}
       \lambda(b')(x_a-y_a).
\]
Fix $b\in B$. By Lemma~\ref{qi-lem:exposure}, some $a_0\in D$
makes $t_0=a_0+b$ uniquely representable as a sum from $D+B$.
The displayed sum at $t_0$ therefore has exactly one term:
\[
  u_{t_0}(x)-u_{t_0}(y)
  =\lambda(b)(x_{a_0}-y_{a_0})\in\{\lambda(b),-\lambda(b)\},
\]
since the two bits differ. Figure~\ref{qi-fig:conv} illustrates
the argument.
\end{proof}

\begin{figure}[t]
\centering
\begin{tikzpicture}[font=\small, x=0.86cm, y=0.86cm]
\newcommand{\cellat}[3]{%
  \draw[#3] (#1-0.5,#2-0.5) rectangle (#1+0.5,#2+0.5);%
}
\begin{scope}
  \foreach \i in {0,1,2} \foreach \j in {0,1,2} {
    \cellat{\i}{\j}{}
    \node[text=gray] at (\i,\j) {$0$};
  }
  \cellat{1}{1}{fill=blk0}
  \node at (1,1) {$+1$};
  \cellat{2}{0}{fill=blk0}
  \node at (2,0) {$+1$};
  \cellat{0}{2}{fill=blk3}
  \node at (0,2) {$-1$};
  \node at (1,-1.05) {$x-y$ on $A$};
  \node at (3.0,1) {and};
\end{scope}
\begin{scope}[xshift=3.45cm]
  \cellat{0}{0}{draw=gray!50, densely dotted}
  \cellat{1}{1}{draw=gray!50, densely dotted}
  \cellat{1}{0}{fill=blk1}
  \node at (1,0) {$\lambda_1$};
  \cellat{0}{1}{fill=blk4}
  \node at (0,1) {$\lambda_2$};
  \node at (0.5,-1.05) {weights $\lambda$ on $B$};
  \node at (2.05,1) {$\longrightarrow$};
\end{scope}
\begin{scope}[xshift=6.65cm]
  \foreach \i in {0,...,3} \foreach \j in {0,...,3} {
    \pgfmathtruncatemacro\out{%
      (\i==0 && \j==0) || (\i==3 && \j==3) ? 1 : 0}
    \ifnum\out=1
      \cellat{\i}{\j}{draw=gray!50, densely dotted}
    \else
      \cellat{\i}{\j}{}
      \node[text=gray] at (\i,\j) {$0$};
    \fi
  }
  \cellat{2}{1}{fill=gray!20}
  \node[font=\scriptsize]
    at (2,1) {$\lambda_1{+}\lambda_2$};
  \cellat{1}{2}{fill=gray!20}
  \node[font=\scriptsize]
    at (1,2) {$\lambda_2{-}\lambda_1$};
  \cellat{3}{0}{fill=good!25, draw=good, very thick}
  \node at (3,0) {$+\lambda_1$};
  \cellat{0}{3}{fill=good!25, draw=good, very thick}
  \node at (0,3) {$-\lambda_2$};
  \node at (1.5,-1.05) {$u_t(x)-u_t(y)$ on $A+B$};
  \node[good!70!black, font=\footnotesize, anchor=west]
    at (3.6,0) {isolates $b_1$};
  \node[good!70!black, font=\footnotesize, anchor=east]
    at (-0.6,3) {isolates $b_2$};
\end{scope}
\end{tikzpicture}
\caption{The weighted sums in Lemma~\ref{qi-lem:measurements},
with $A=\{0,1,2\}^2$ and $B=\{b_1,b_2\}=\{(1,0),(0,1)\}$.
Write $\lambda_i=\lambda(b_i)$. Each contribution
$\lambda(b)(x_a-y_a)$ is placed at $a+b$. Some contributions meet
in the gray cells, but the two green cells each contain a single
contribution, giving $+\lambda_1$ and $-\lambda_2$.
In the right panel, dotted cells lie outside $A+B$.}
\label{qi-fig:conv}
\end{figure}

We now use each weighted sum and its negative as shifts, together
with one zero shift.

\begin{lemma}\label{qi-lem:partitions}
Let $p\ge3$ and $\kappa=\lfloor p/2\rfloor$. Suppose
$A,B\subseteq\Z^r$ are finite and nonempty, $|B|=\kappa$, and
every point of $B$ uniquely maximizes some linear functional on
$B$. There are $2^{|A|}$ pairwise qualitatively independent
partitions into $p$ blocks, with a common partition into
$2|A+B|+1$ sets of size $p$, such that every block contains exactly
one point of each common set.
\end{lemma}

\begin{proof}
Choose a bijection $\lambda:B\to\{1,\ldots,\kappa\}$ and form
the sums $u_t(x)$ from Lemma~\ref{qi-lem:measurements}. Use two
indices $(t,+)$ and $(t,-)$ for each $t\in A+B$, and one extra
index $j_0$, with shifts
\[
  h_{(t,+)}(x)=u_t(x),\qquad
  h_{(t,-)}(x)=-u_t(x),\qquad
  h_{j_0}(x)=0.
\]
For distinct $x,y$, Lemma~\ref{qi-lem:measurements} ensures that
these differences include both $\lambda(b)$ and $-\lambda(b)$
for every $b\in B$, hence every nonzero residue. The extra index
supplies zero. Lemma~\ref{qi-lem:shifts}, with
$X=\{0,1\}^A$, gives the required partitions.
\end{proof}

\subsection{Choice of parameters}\label{qi-sec:choice}
\label{qi-sec:asymptotics}

For positive integers $M,k$, we take $A=\{0,\ldots,M-1\}^r$
and choose $B$ from points of equal norm in $\{0,\ldots,k\}^r$.
This ensures the geometric condition on $B$ while keeping $A+B$
inside $\{0,\ldots,M+k-1\}^r$
(Figure~\ref{qi-fig:boxes}).

\begin{center}
\begin{minipage}{\linewidth}
\captionsetup{type=figure}
\centering
\begin{tikzpicture}[
  font=\small,
  x=0.27cm,
  y=0.27cm,
  ex/.style={
    circle,
    fill=black,
    inner sep=0pt,
    minimum size=3.6pt
  },
  lat/.style={
    circle,
    fill=black!55,
    inner sep=0pt,
    minimum size=1.8pt
  }
]
  \fill[blk0!60] (0,0) rectangle (3,3);
  \draw[blk0!50!black] (0,0) rectangle (3,3);
  \foreach \i in {0,...,3} \foreach \j in {0,...,3} {
    \node[lat] at (\i,\j) {};
  }
  \node[font=\footnotesize] at (1.5,-2.2) {$A$};
  \node at (4.6,2.5) {$+$};

  \begin{scope}[xshift=1.7cm]
    \draw[gray!70] (0,0) rectangle (5,5);
    \draw[gray] (5,0)
      arc[start angle=0, end angle=90, radius=5];
    \foreach \a/\b in {0/5,3/4,4/3,5/0} {
      \node[ex] at (\a,\b) {};
    }
    \node[font=\footnotesize] at (2.5,-2.2) {$B$};
    \node at (6.6,2.5) {$=$};
  \end{scope}

  \begin{scope}[xshift=4.0cm]
    \draw[densely dashed] (0,0) rectangle (8,8);
    \foreach \a/\b in {0/5,3/4,4/3,5/0} {
      \fill[blk0, fill opacity=0.5]
        (\a,\b) rectangle (\a+3,\b+3);
    }
    \foreach \a/\b in {0/5,3/4,4/3,5/0} {
      \draw[blk0!50!black]
        (\a,\b) rectangle (\a+3,\b+3);
      \node[ex] at (\a,\b) {};
    }
    \node[font=\footnotesize] at (4,-2.2) {$A+B$};
  \end{scope}
\end{tikzpicture}
\caption{For positive integers $M,k$, let
$A=\{0,\ldots,M-1\}^r$ and $B\subseteq\{0,\ldots,k\}^r$.
The sumset $A+B$ is the union of the translates $A+b$, $b\in B$,
and is contained in $\{0,\ldots,M+k-1\}^r$.}
\label{qi-fig:boxes}
\end{minipage}
\end{center}

\begin{proposition}\label{qi-prop:finite}
Let $p\ge3$ and $N\ge2$ be integers, and put
$m=\lceil\log_2N\rceil$ and $\kappa=\lfloor p/2\rfloor$.
For every integer $r\ge3$, there are $N$ pairwise qualitatively
independent partitions into $p$ blocks and a common partition into
$\ell$ sets of size $p$, with exactly one point of each common set
in every block, where
\begin{equation}\label{qi-eq:finite}
  \ell\le
  2\left(\left\lceil m^{1/r}\right\rceil
       +\left\lceil(2r\kappa)^{1/(r-2)}\right\rceil\right)^r+1.
\end{equation}
\end{proposition}

\begin{proof}
Set $M=\lceil m^{1/r}\rceil$ and
$k=\lceil(2r\kappa)^{1/(r-2)}\rceil$, and take
$A=\{0,\ldots,M-1\}^r$, so $|A|\ge m$.
The squared norms of points in $\{0,\ldots,k\}^r$ take at most
$rk^2+1$ values. Some set of points with the same norm therefore
has size at least
\[
  \frac{(k+1)^r}{rk^2+1}
  \ge\frac{k^{r-2}}{2r}\ge\kappa.
\]
Choose any $\kappa$ of these points as $B$.
Lemma~\ref{qi-lem:sphere} verifies the geometric condition, and
\[
  A+B\subseteq\{0,\ldots,M+k-1\}^r,
  \qquad |A+B|\le(M+k)^r.
\]
Lemma~\ref{qi-lem:partitions} gives at least $2^m\ge N$
partitions using $\ell\le2(M+k)^r+1$ common sets. Retain any $N$
of them to obtain~\eqref{qi-eq:finite}.
\end{proof}

\begin{proof}[Proof of Theorem~\ref{thm:qi}]
For sufficiently large $p$, put $L=\ln p$ and
$r=\lceil\sqrt L\rceil\ge4$, and use the choices of $m,M,k$
above. Uniformly for $p^p\le N\le p^{3p}$,
$p\le m\le3p\log_2p+1$. Thus
\[
  M^r\le m\bigl(1+m^{-1/r}\bigr)^r
       \le m\exp\bigl(rm^{-1/r}\bigr)=O(m),
\]
since $rm^{-1/r}\le r e^{-L/r}=o(1)$.
Also, $2r\kappa\le rp$, $k\le2(2r\kappa)^{1/(r-2)}$, and
$M\ge e^{L/r}$, so
\[
  \ln\frac{k}{M}
  \le\ln2+\frac{L+\ln r}{r-2}-\frac Lr
  =\ln2+\frac{2L}{r(r-2)}+\frac{\ln r}{r-2}=O(1).
\]
It follows that
\[
  (M+k)^r=M^r(1+k/M)^r\le O(m)e^{O(r)}.
\]
Proposition~\ref{qi-prop:finite} therefore gives
\[
  \ell\le2(M+k)^r+1
  \le m\,e^{O(r)}
  \le p\,e^{O(\sqrt{\ln p})},
\]
where the last inequality uses $m=O(p\log p)$ and
$\log p=e^{o(\sqrt{\ln p})}$.
The ground set has at most $p^2e^{O(\sqrt{\ln p})}$ points, as
required.
\end{proof}

\subsection{From partitions to prophet inequalities}

\begin{proof}[Proof of Theorem~\ref{thm:iid}]
We recall the standard reduction of~\cite{FPT25,SVW23}.
Apply Theorem~\ref{thm:qi} with $N=p^p$, and write
$(C_j)_{j\in J}$ for the common partition, where
$\ell=|J|\le p e^{O(\sqrt{\log p})}$.
Introduce one element $(x,s)$ for each block $P_s^x$ of each
partition $\mathcal P^x$. Each $C_j$ defines a unit-capacity
partition matroid: group elements according to the unique point of
$P_s^x\cap C_j$. Transversality makes these groups a partition of
the elements. Blocks within one $\mathcal P^x$ are disjoint,
whereas blocks from different partitions intersect by qualitative
independence. Thus the common independent sets of these $\ell$
matroids are exactly the subsets of one group
$G_x=\{(x,s):s\in\Z_p\}$.

Give all elements i.i.d.\ Bernoulli values of mean $1/p$.
In every arrival order, an online algorithm's first acceptance
commits it to one group; the unrevealed values in that group have
total conditional expectation at most $(p-1)/p$, so $\ON<2$.
The prophet obtains value $p$ whenever some group is entirely
positive, and hence
\[
  \E[\OPT]\ge p\bigl(1-(1-p^{-p})^{p^p}\bigr)
  \ge(1-e^{-1})p.
\]
Consequently, $\Gap=\Omega(p)$.

Finally, for sufficiently large $q$, take
$p=\lfloor q e^{-C\sqrt{\log q}}\rfloor$ with a sufficiently
large absolute constant $C$. Then $\ell\le q$ and
$p=q e^{-O(\sqrt{\log q})}$. Adding redundant partition matroids
with singleton blocks gives exactly $q$ matroids and proves
Theorem~\ref{thm:iid}; increasing the implicit constant covers
the remaining values of $q$.
\end{proof}

 

\section{Conclusion}
\label{sec:conclusion}

We establish an \(\Omega(q/\log q)\) lower bound for prophet inequalities under intersections of \(q\) partition matroids. In the single-minded formulation, our construction yields an \(\Omega(d/\log d)\) lower bound for buyers requesting fixed bundles of at most \(d\) unit-capacity items. These results show that the known linear upper bounds are optimal up to a logarithmic factor. Our construction uses independent two-point distributions whose positive values vary across buyers. Closing the remaining logarithmic gap is open.

For i.i.d. Bernoulli values, we obtain a nearly linear lower bound of \(qe^{-O(\sqrt{\log q})}\), valid for every arrival order. The underlying explicit construction of qualitatively independent partitions also shows that the product dimension of \(p^p\) disjoint cliques of size \(p\) is \(p^{1+o(1)}\). A natural next question is whether the \(\Omega(q/\log q)\) lower bound can also be attained with i.i.d. Bernoulli values.

\section*{AI Methodology}
The construction behind Theorems~\ref{thm:partition-lower-bound} and~\ref{thm:single-minded-lower-bound} was found during extended discussions with ChatGPT 6 Astra. More specifically, we guided the model toward adapting the construction of \cite{RS} to a construction whose feasibility constraints are expressed through items and demand bundles of single-minded buyers. We supplied a blueprint for how we expected the lower-bound instance to be and continued to refine the final construction through successive exchanges with the model. The construction behind Theorem~\ref{thm:iid} was found during extended interaction with ChatGPT 6 Astra, where the authors asked the model to evaluate a small number of different alternatives. We subsequently rewrote and simplified the proof. We also used Claude Opus 5.5 and ChatGPT 6 Astra for assistance with the organization and the presentation of this manuscript. The authors verified the correctness and originality of all content, including references, and assume full responsibility for the content.

\printbibliography

\end{document}